\documentclass[12pt,a4paper]{article}
\usepackage[utf8]{inputenc}
\usepackage[margin=1in]{geometry}
\usepackage{amsmath}
\usepackage{amssymb}
\usepackage{amsthm}
\usepackage{microtype}
\usepackage{tcolorbox}
\usepackage{booktabs}
\usepackage{array}
\usepackage{graphicx}
\usepackage{cite}
\usepackage{hyperref}

\newcommand{\order}[1]{\mathcal{O}(#1)}
\newcommand{\abs}[1]{\left|#1\right|}

\newtheorem{theorem}{Theorem}[section]
\newtheorem{proposition}[theorem]{Proposition}
\theoremstyle{definition}
\newtheorem{postulate}[theorem]{Postulate}
\theoremstyle{remark}
\newtheorem{remark}[theorem]{Remark}

\title{\textbf{Logarithmic Spectral Phase Deformations:\\Observable Signatures and the Limits\\of Spontaneous Dephasing}}
\author{Sridhar Tayur\\
Carnegie Mellon University, Pittsburgh, PA 15213\\
\texttt{stayur@cmu.edu}}
\date{Revised September 23, 2026}

\begin{document}

\maketitle

\begin{abstract}
We study a unitary deformation of the relation between a positive self-adjoint reference Hamiltonian $H$ and the generator of phase evolution,
\begin{equation*}
F_\beta(H) = H + \beta H \log(H/E_*).
\end{equation*}
Since $F_\beta(H)$ is a real function of $H$, it has the same eigenvectors, and it preserves unitarity, purity and the Born rule. It does not cause decoherence in the sense of a loss of purity. Its effect is a nonlinear distortion of the spectrum. Changing the reference scale $E_*$ adds a term linear in $H$, which cannot be distinguished from a rescaling of the spectrum. As a result, a single transition frequency does not determine $\beta$, and neither does the static dephasing of an ensemble; a finite-width Ramsey model makes this explicit. The part of the deformation that does not depend on $E_*$ is contained in second and higher spectral differences, which sample the curvature of the spectral map. For a nominally harmonic ladder $E_n=\omega(n+\delta)$, the adjacent-transition anharmonicity is $A_n/\omega=\beta c_n(\delta)$, with $c_n\sim(n+\delta)^{-1}$. This decreasing template can be distinguished in a joint fit from ideal Kerr and specified polynomial nonlinearities using sufficiently many resolved levels, subject to platform-specific modeling of other effects. We describe a spectroscopy protocol, estimate the precision of a joint fit, and show a synthetic example. The result depends on the choice of energy origin, which second differences do not remove, and it applies only to a single system during free evolution, because the deformation is not additive over subsystems.
\end{abstract}

\tableofcontents

\section{Introduction}

In standard quantum mechanics the Hamiltonian $H$ has two roles. It gives the energies of a system, as inferred from spectroscopy or assigned by a model, and it generates time translation, $U(t)=e^{-iHt}$. In this paper we ask what happens if the two roles are separated by a fixed nonlinear map.

\subsection{The ansatz}

We take $H$ to be a reference operator that fixes the energy labels of a system under a stated convention, and we take the generator of phase evolution to be
\begin{equation}
F_\beta(H) = H + \beta H \log \frac{H}{E_*},
\label{eq:deformation}
\end{equation}
defined through the spectral theorem for $H\ge 0$ and a reference scale $E_*>0$. In standard quantum mechanics the two operators coincide. Energy eigenstates evolve as
\begin{equation}
e^{-itF_\beta(H)}|E\rangle = \exp\left[-itE \left(1 + \beta \log \frac{E}{E_*}\right)\right] |E\rangle.
\label{eq:eigenstate}
\end{equation}

As an operator, $F_\beta(H)$ is just another self-adjoint Hamiltonian. The proposal has content only if the map $H\mapsto F_\beta(H)$ is taken to be universal, with one value of $\beta$ for all systems. It can then be tested wherever $H$ is known independently of the measured spectrum, for example when a symmetry fixes the levels to be equally spaced. If $H$ is inferred from the same spectrum that is being measured, the deformation cannot be seen, since any spectrum can be relabeled. The ansatz is therefore not an alternative parametrization of an unknown Hamiltonian. It has empirical content only where the reference Hamiltonian is fixed by independent structure or calibration before the deformed spectrum is measured.

Because $F_\beta(H)$ commutes with $H$, the evolution is unitary, purity is conserved, and $|\langle E_m|\rho(t)|E_n\rangle|$ is constant for sharp levels. Coherence is lost only after averaging over a spread of splittings, over a wave packet, or over finite measurement resolution, as in ordinary inhomogeneous broadening. The model is therefore a modification of the energy--phase relation and should not be compared with collapse or open-system models as a theory of decoherence.

\subsection{A poetic intuition}

Wheeler said that space tells matter how to move, and matter tells space how to curve. By analogy:
\begin{quote}
energy tells time how to tick, and time tells energy how to dephase.
\end{quote}
This is not meant literally, since time is still an external parameter. Under the ansatz, a component with energy $E$ accumulates phase as $tE(1+\beta\log E)$ instead of $tE$, so different energies ``experience time'' at slightly different rates. Think of an ensemble of clocks whose rates depend a little on their energies. Each clock keeps perfect time. An observable that averages over clocks with a spread of energies sees their drift as a loss of coherence, but because the drift is deterministic and static it can in principle be reversed, as a spin echo reverses inhomogeneous dephasing. This paper is about the limits of this picture, and what the drift can and cannot tell us about $\beta$.

\subsection{Status of the proposal}

Table~\ref{tab:status} summarizes what is and is not established.

\begin{table}[h]
\centering
\small
\begin{tabular}{@{}>{\raggedright}p{3.2cm}>{\raggedright\arraybackslash}p{11.2cm}@{}}
\toprule
\textbf{Level} & \textbf{Status} \\
\midrule
Mathematical model & Well defined as a functional-calculus map for a given positive operator $H$. \\
Phenomenology & Testable once an energy origin is chosen, and only for a single system during free evolution. \\
Fundamental theory & Not established. The energy origin and the rule for composite systems would have to follow from a deeper, preferably covariant, construction. \\
\bottomrule
\end{tabular}
\caption{Levels at which the logarithmic deformation can be read.}
\label{tab:status}
\end{table}

The most serious problem is the energy origin. In ordinary quantum mechanics, $H\mapsto H+c$ changes only a global phase. Here it changes relative phases, because $F_\beta(E+c)-F_\beta(E)$ is not affine in $E$. The model is therefore incomplete unless some further framework fixes the origin. We treat the origin as an explicit convention and show how each prediction depends on it (Section~\ref{sec:origin}).

\subsection{Results}

Changing $E_*$ adds a term linear in $H$, which is equivalent to rescaling the reference spectrum, so single transition frequencies and the static dephasing they cause do not determine $\beta$ (Sections~\ref{sec:Estar} and~\ref{sec:ramsey}). Second and higher divided differences of the spectrum do not depend on $E_*$, though they do depend on the energy origin (Section~\ref{sec:invariants}). In a nominally harmonic ladder the deformation gives an anharmonicity $A_n/\omega=\beta c_n(\delta)\approx\beta/(n+\delta)$, which decreases with $n$ and can be jointly fitted against Kerr and specified polynomial terms when enough levels are measured (Sections~\ref{sec:ladder} and~\ref{sec:experiment}).

An earlier version of this paper gave bounds on $\beta$ from qubit coherence times and Ramsey frequencies. Those bounds were wrong and are withdrawn; the reasons are given in Section~\ref{sec:withdrawn}.

\section{Mathematical Structure}
\label{sec:math}

\subsection{Definition, self-adjointness, unitarity}

Let $H\ge 0$ be self-adjoint with projection-valued measure $\Pi$, so $H=\int_0^\infty E\,d\Pi(E)$. Positivity is needed for the logarithm as written. A Hamiltonian that is not positive must first be shifted, and that shift is the choice of energy origin discussed in Section~\ref{sec:origin}. Define
\begin{equation}
F_\beta(E) = E\left(1 + \beta \log \frac{E}{E_*}\right), \qquad F_\beta(0):=0,
\label{eq:Fbeta}
\end{equation}
which is continuous on $[0,\infty)$, and set $F_\beta(H)=\int F_\beta(E)\,d\Pi(E)$.

\begin{proposition}
For every $\beta\in\mathbb{R}$, $F_\beta(H)$ is self-adjoint on the maximal domain
\begin{equation}
D(F_\beta(H)) = \left\{\psi : \int_0^\infty |F_\beta(E)|^2 \, d\langle\psi, \Pi(E)\psi\rangle < \infty\right\},
\end{equation}
and $U_\beta(t)=e^{-itF_\beta(H)}$ is a strongly continuous one-parameter unitary group.
\end{proposition}

\begin{proof}
$F_\beta$ is real and Borel, so the spectral theorem gives self-adjointness on the maximal domain. The domain is dense: $F_\beta$ is bounded on each $[0,N]$, so $\Pi([0,N])\mathcal{H}\subset D(F_\beta(H))$, and $\Pi([0,N])\to I$ strongly as $N\to\infty$. Stone's theorem gives the unitary group.
\end{proof}

\begin{remark}[Level ordering]
\label{rem:window}
Self-adjointness and unitarity hold for all real $\beta$ and do not require $F_\beta$ to be monotone. Monotonicity matters only if $F_\beta$ is to preserve the ordering of levels. Since $F'_\beta(E)=1+\beta[\log(E/E_*)+1]$ vanishes at $E_c=E_*\exp(-1-1/\beta)$, the ordering is preserved for $E_*e^{-1/|\beta|}\ll E\ll E_*e^{1/|\beta|}$. For $\beta<0$, $F_\beta$ is also negative above $E_*e^{1/|\beta|}$. When $|\beta|\ll 1$ this range is very wide, and we work within it.
\end{remark}

\subsection{Reference-scale redundancy}
\label{sec:Estar}

Write $G(E):=E\log(E/E_*)$, so that $F_\beta = E+\beta G$. Replacing $E_*$ by $E_*'$ gives $G\to G+E\log(E_*/E_*')$, and hence
\begin{equation}
F_\beta(H)\big|_{E_*'} = \left(1 + \beta\log\frac{E_*}{E_*'}\right) H + \beta H\log\frac{H}{E_*}.
\label{eq:Estarshift}
\end{equation}
A change of $E_*$ therefore rescales $H$; it does not shift its zero. First differences $G(E_m)-G(E_n)$ depend on $E_*$ only through a term proportional to $E_m-E_n$, which has the same effect as a change in the normalization of the reference spectrum or in the fitted parameters of a model.

\subsection{Energy-origin dependence}
\label{sec:origin}

A separate ambiguity is the origin, $E\to E+c$. The two kinds of change behave differently. A rescaling $E\mapsto aE$ of the energy unit or of the reference spectrum gives $F_\beta(aE)=a\,F_\beta(E)+\beta aE\log a$. The additional term $\beta aE\log a$ is linear in $E$, so it is degenerate with a rescaling of the reference spectrum; equivalently, it is absorbed by transforming the energy scale and $E_*$ together. Normalized ladder quantities such as $A_n/\omega$ (Section~\ref{sec:ladder}) are therefore independent of the overall energy scale. An additive shift $E\mapsto E+b$ cannot be absorbed in this way: it moves the point at which the curvature $G''(E)=1/E$ is sampled, and in the ladder it changes the dimensionless offset $\delta$ defined in Eq.~\eqref{eq:ladderspec} and hence the shape of the predicted signal.

Every prediction below therefore depends on the origin, and we write them in terms of an explicit offset. Two conventions illustrate the problem. In a subsystem convention, the offset is stipulated through a chosen representation of the system's Hamiltonian, for example $\delta=\tfrac12$ for the conventional oscillator Hamiltonian or $\delta=0$ after normal ordering. In the absolute convention, rest energies are included, $E\to Mc^2+E$; the curvature at laboratory transition energies is then suppressed by $\omega/Mc^2$ and the effects discussed below are far too small to observe. In standard quantum mechanics, normal ordering and similar choices are conventions. Under this ansatz they have physical consequences, and the ansatz gives no rule for choosing among them. We regard this as the main limitation of the constant-$\beta$ model.

\subsection{Scope: a single system during free evolution}
\label{sec:scope}

For noninteracting subsystems, $H_{AB}=H_A\otimes I+I\otimes H_B$,
\begin{equation}
\begin{aligned}
&F_\beta(H_{AB}) - F_\beta(H_A)\otimes I - I\otimes F_\beta(H_B)\\
&\qquad= \beta\big[H_{AB}\log H_{AB} - H_A\log H_A\otimes I - I\otimes H_B\log H_B\big]\neq 0,
\end{aligned}
\label{eq:composite}
\end{equation}
so applying the deformation to the total Hamiltonian couples systems that are otherwise independent. If the factorization is fixed and the deformation is applied to each factor separately, $F_\beta(H_A)\otimes I+I\otimes F_\beta(H_B)$, initially factorized states stay factorized under free evolution, and within this restricted prescription no influence passes between the subsystems. However, what counts as an isolated subsystem depends on how the system is divided up. This rule also does not say what the generator becomes when an interaction $V$ is switched on (if it is $F_\beta(H_A+H_B+V)$, it jumps at $V=0$), how identical particles or field modes should be divided, or whether different divisions give different predictions. We do not answer these questions here, and we restrict the analysis as follows.

\begin{postulate}
\label{post:scope}
The analysis applies to a single system during an interval of approximately free evolution, with a known reference Hamiltonian $H$ with discrete spectrum and a stated choice of energy origin. During that interval its evolution is generated by $F_\beta(H)$. Preparation, control and readout act outside the interval; their effects must either be absent during it or be modeled separately. No claim is made about composite or interacting systems.
\end{postulate}

This is enough for spectroscopy of a single mode, which is the test we propose.

\section{Spectral Invariants}
\label{sec:invariants}

\subsection{Two-level systems}

For two levels, the transition frequency becomes $\omega_{\text{eff}}=(E_2-E_1)+\beta[G(E_2)-G(E_1)]$. With $E_1=2\pi\times5$~GHz, $E_2=2\pi\times5.1$~GHz and $E_*=2\pi\times1$~GHz, we get $G(E_2)-G(E_1)\approx1.65\times10^9\,\text{s}^{-1}\approx\Delta E\,[\log(E/E_*)+1]$ with $\log(E/E_*)+1\approx2.62$, so $\beta=10^{-6}$ shifts the transition by about 260~Hz. The factor 2.62 depends on $E_*$, and by Eq.~\eqref{eq:Estarshift} the shift is absorbed into the fitted value of the splitting. A single transition can constrain $\beta$ only if the undeformed splitting is known independently to sufficient precision, which is generally not the case for engineered systems such as superconducting qubits.

\subsection{Divided differences}

For three levels $E_1<E_2<E_3$, the second divided difference is
\begin{equation}
F_\beta[E_1,E_2,E_3] = \sum_{i=1}^3 \frac{F_\beta(E_i)}{\prod_{j\neq i}(E_i-E_j)} = \beta\,G[E_1,E_2,E_3].
\label{eq:divdiff}
\end{equation}
By the mean-value theorem for divided differences, there is a $\xi\in(E_1,E_3)$ with
\begin{equation}
G[E_1,E_2,E_3]=\tfrac12 G''(\xi)=\frac{1}{2\xi}.
\end{equation}
The second divided difference vanishes for any linear function of $E$ and so does not depend on $E_*$, for any three levels. For equally spaced levels, $E_1-2E_2+E_3=0$, it equals the ordinary second difference divided by $2(\Delta E)^2$. Higher divided differences have the same property. The $E_*$-independent content of the deformation is thus contained in second and higher divided differences, which sample the curvature $G''(E)=1/E$. Their values still depend on the chosen energy origin, through $\xi$.

\subsection{Harmonic ladder}
\label{sec:ladder}

Consider a nominally harmonic spectrum with offset $\delta\ge0$,
\begin{equation}
E_n=\omega(n+\delta), \qquad n=0,1,2,\ldots,
\label{eq:ladderspec}
\end{equation}
where $\delta$ encodes the choice of origin: $\delta=\tfrac12$ with the zero-point energy included, $\delta=0$ with normal ordering, and $\delta\sim Mc^2/\omega$ in the absolute convention. Define transition frequencies $\omega_n:=F_\beta(E_{n+1})-F_\beta(E_n)$ for $n\ge0$ and anharmonicities $A_n:=\omega_n-\omega_{n-1}$ for $n\ge1$. With $g(x)=x\log x$ and $g(0)=0$, the terms involving $E_*$ and $\log\omega$ cancel, and
\begin{equation}
\frac{A_n}{\omega} = \beta\, c_n(\delta), \qquad
c_n(\delta) = g(n+1+\delta) - 2g(n+\delta) + g(n-1+\delta), \qquad n\ge1,
\label{eq:ladder}
\end{equation}
with $c_n(\delta)\to 1/(n+\delta)$ for $n+\delta\gg1$. Table~\ref{tab:cn} gives some values.

\begin{table}[h]
\centering
\small
\begin{tabular}{@{}cccccc@{}}
\toprule
$\delta$ & $c_1$ & $c_2$ & $c_3$ & $c_5$ & $c_{10}$ \\
\midrule
$0$ & 1.386 & 0.523 & 0.340 & 0.201 & 0.100 \\
$1/2$ & 0.728 & 0.411 & 0.290 & 0.183 & 0.095 \\
$1$ & 0.523 & 0.340 & 0.253 & 0.167 & 0.091 \\
$5$ & 0.167 & 0.143 & 0.125 & 0.100 & 0.067 \\
$20$ & 0.048 & 0.045 & 0.044 & 0.040 & 0.033 \\
\bottomrule
\end{tabular}
\caption{Ladder coefficients $c_n(\delta)$ of Eq.~\eqref{eq:ladder}.}
\label{tab:cn}
\end{table}

For $\delta$ of order one, the choice of origin changes $c_n$ by a factor of order one, and $c_n$ still decreases with $n$. For $\delta\gg n$, $c_n\approx1/\delta$ is nearly constant and cannot be told apart from a Kerr term, so the test fails as the absolute convention is approached. An ideal Kerr Hamiltonian $\omega\hat n+\tfrac{K}{2}\hat n(\hat n-1)$ gives $A_n=K$ for all $n$, and polynomial nonlinearities give $A_n$ that is constant or grows with $n$. Real resonators also have dressing, higher-order terms, dispersive and drive-induced shifts, and measurement backaction, and these have to be modeled (Section~\ref{sec:systematics}). Some of these mechanisms, such as saturation, mode hybridization or non-polynomial circuit potentials, can also produce an effective anharmonicity that decreases with $n$. The profile $c_n(\delta)$ is therefore not unique to the logarithmic deformation. It is a constrained template to be fitted jointly with a model of the platform, not a model-independent signature.

\subsection{Other spectra}

For unequally spaced levels, the ordinary second difference in the level index $n$ is not the right quantity; the divided difference in energy, Eq.~\eqref{eq:divdiff}, is. The quartic oscillator, with $E_n\sim\kappa n^{4/3}$ for large $n$, shows why. At leading order in large $n$, the $\beta$-dependent part of the second difference in $n$, with $\Delta_n^2 f_n := f_{n+1}-2f_n+f_{n-1}$, is
\begin{equation}
\beta\,\Delta_n^2 G(E_n) \approx \beta\left[G'(E_n)\,E_n'' + G''(E_n)\,(E_n')^2\right]
= \beta\kappa n^{-2/3}\left[\tfrac49\left(\log\frac{E_n}{E_*}+1\right)+\tfrac{16}{9}\right],
\end{equation}
which grows like $n^{-2/3}\log n$ and depends on $E_*$ through the term proportional to the undeformed second difference $E_n''$. The divided difference removes that term and gives $\beta G[E_{n-1},E_n,E_{n+1}]\approx\beta/(2E_n)\sim\beta/(2\kappa n^{4/3})$. Using it, however, requires the undeformed $E_n$ to be known to high precision, and since the reference spectrum is already anharmonic and model dependent, the quartic oscillator is less useful than the harmonic one. Minisuperspace and black-hole-interior toy models that are harmonic near a minimum reduce to Section~\ref{sec:ladder}, although for Hamiltonians constrained to vanish the choice of origin is even harder to justify.

\section{Static Ensemble Dephasing: A Finite-Width Ramsey Model}
\label{sec:ramsey}

To see why coherence measurements do not determine $\beta$, consider a two-level system with $E_0=0$ whose splitting $E_1(\lambda)$ depends on a control parameter $\lambda$ (a flux, field or trap voltage). Suppose $\lambda$ varies from shot to shot with a Gaussian distribution of width $\sigma_\lambda$ about $\lambda_0$. The Ramsey frequency is
\begin{equation}
\omega_\beta(\lambda) = E_1(\lambda)\left[1+\beta\log\frac{E_1(\lambda)}{E_*}\right].
\end{equation}
With $s=dE_1/d\lambda$ at $\lambda_0$,
\begin{equation}
\frac{d\omega_\beta}{d\lambda} = s\left[1+\beta\left(\log\frac{E_1}{E_*}+1\right)\right] =: \tilde s,
\end{equation}
and the ensemble-averaged coherence decays as
\begin{equation}
|\chi(t)| = \exp\!\left(-\tfrac12\,\tilde s^{\,2}\sigma_\lambda^2 t^2\right).
\label{eq:ramsey}
\end{equation}
The decay is Gaussian in $t$, not exponential. It comes from a static spread of frequencies and is therefore removed by a spin echo \cite{Hahn1950}, so echo times $T_2$ carry no information about $\beta$. Compared with $\beta=0$, the only change is $s\to\tilde s$. In practice $s$ is calibrated by measuring $\omega_\beta$ as a function of $\lambda$, and this measurement returns $\tilde s$. The envelope predicted with the calibrated slope is then the same as for $\beta=0$. To extract $\beta$ one would need an independent value of $s$, accurate to a relative precision $|\beta(\log(E_1/E_*)+1)|$, and this quantity itself depends on $E_*$. At second order, $d^2\omega_\beta/d\lambda^2$ contains a term $\beta s^2/E_1$ that does not depend on $E_*$, but it cannot be separated from the curvature of $E_1(\lambda)$ unless that curvature is known independently.

More generally, for an energy distribution with a smooth density, $\int f(E)e^{-itF_\beta(E)}dE$ decays as $\order{1/|t|}$ whether or not $\beta=0$; the deformation changes only the slope of the phase. Appendix~\ref{app:theorem} gives a precise statement.

\section{Experimental Tests}
\label{sec:experiment}

The test depends on knowing the reference spectrum independently of the spectroscopy used to measure $A_n$. Observing that a mode is nearly harmonic at low excitation does not by itself distinguish $H$ from $F_\beta(H)$. A meaningful test needs a platform in which the quadratic reference Hamiltonian follows from a separately calibrated model, such as the electrode geometry and applied potentials of an ion trap or the geometry and materials of a linear cavity, and in which conventional departures from harmonicity are bounded by independent measurements or microscopic calculations. The fit then tests for a universal residual component proportional to $c_n(\delta)$ once those conventional contributions have been specified. If the conventional anharmonicity is instead left fully free, the deformation can be absorbed into it.

\subsection{Protocol}

\begin{tcolorbox}[colback=gray!5!white,colframe=gray!75!black,title=Harmonic-ladder protocol]
\begin{enumerate}
\item Choose a single mode whose reference spectrum is known independently to be harmonic, such as a trapped-ion motional mode or a linear microwave or optical cavity mode, and state the choice of origin $\delta$.
\item Measure the transition frequencies $\omega_n$ ($|n\rangle\to|n+1\rangle$) for $n=0,\ldots,N$, for example by sideband or photon-number-resolved spectroscopy.
\item Form $A_n=\omega_n-\omega_{n-1}$ for $n=1,\ldots,N$.
\item Fit $A_n/\omega=\beta c_n(\delta)+K+K'n+\cdots$, where $K$ is the Kerr coefficient and $K'$ the next polynomial term, together with any systematic shifts that can be modeled independently.
\item Report $\beta$, or a bound on it, as a function of $\delta$ over the range of $\delta$ that can be justified for the platform.
\end{enumerate}
\end{tcolorbox}

\subsection{Precision of the fit}

For fixed $\delta$ the fit is linear in $(\beta,K,K',\ldots)$, and its precision is set by the design matrix with columns $c_n(\delta)$, $1$, $n,\ldots$. The measured quantities are the transition frequencies $\omega_n$, and since $A_n=\omega_n-\omega_{n-1}$, errors in adjacent $A_n$ are correlated. We therefore assume independent Gaussian errors $\sigma_\omega$ on each $\omega_n/\omega$, $n=0,\ldots,N$, and fit the measured $\omega_n$ directly with a free overall frequency, or equivalently fit the derived $A_n$ by generalized least squares with covariance $\sigma_\omega^2 DD^{\mathsf T}$, where $D$ is the first-difference matrix. Table~\ref{tab:fisher} gives the resulting standard error on $\beta$ for $\delta=\tfrac12$.

\begin{table}[h]
\centering
\small
\begin{tabular}{@{}lccc@{}}
\toprule
Fit parameters & $N=5$ & $N=10$ & $N=20$ \\
\midrule
$\beta$ only & $0.66\,\sigma_\omega$ & $0.42\,\sigma_\omega$ & $0.27\,\sigma_\omega$ \\
$\beta, K$ & $2.52\,\sigma_\omega$ & $1.29\,\sigma_\omega$ & $0.73\,\sigma_\omega$ \\
$\beta, K, K'$ & $8.36\,\sigma_\omega$ & $3.19\,\sigma_\omega$ & $1.51\,\sigma_\omega$ \\
\bottomrule
\end{tabular}
\caption{Standard error on $\beta$ from a fit of $N$ anharmonicities. Assumptions: independent Gaussian errors $\sigma_\omega$ on each $\omega_n/\omega$ ($n=0,\ldots,N$), with the resulting correlations between adjacent $A_n$ included; fixed $\delta=\tfrac12$; fit basis $\{c_n,1,n\}$ or the indicated subset.}
\label{tab:fisher}
\end{table}

When the conventional nonlinearities are fitted at the same time, several levels are needed, because with only a few points the decreasing function $c_n$ is hard to distinguish from a constant plus a linear term. Figure~\ref{fig:ladder} shows a synthetic example with $\beta=10^{-6}$, $K/\omega=-5\times10^{-7}$ and $\sigma_\omega=10^{-7}$ for $N=15$. A joint fit of $(\beta,K,K')$ gives $\beta=(0.89\pm0.20)\times10^{-6}$. The offset $\delta$ enters nonlinearly, and for small $\beta$ only through terms of order $\beta\,\partial_\delta c_n$, so it cannot be determined from the data. Results have to be stated for a given $\delta$; the right panel of Figure~\ref{fig:ladder} shows how much this matters.

\begin{figure}[h]
\centering
\includegraphics[width=\textwidth]{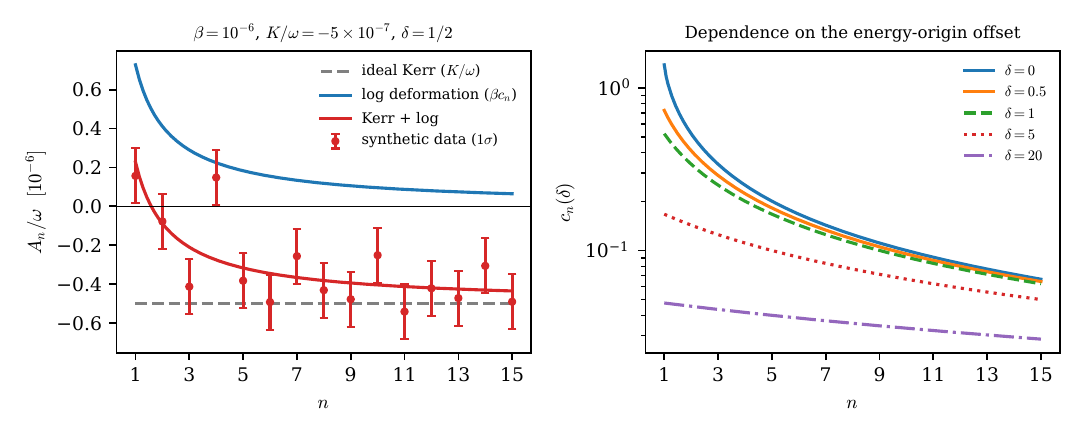}
\caption{Left: anharmonicity $A_n/\omega$ against $n$ for an ideal Kerr term, the logarithmic deformation and their sum, with synthetic data. The data are 15 anharmonicities formed from 16 simulated transition frequencies $\omega_n/\omega$, each with independent Gaussian noise $\sigma_\omega=10^{-7}$; error bars are $1\sigma$ ($\sqrt2\,\sigma_\omega$), and adjacent points are correlated. Right: $c_n(\delta)$ for several offsets. For $\delta\gg n$ the curve is nearly flat and cannot be distinguished from a Kerr term.}
\label{fig:ladder}
\end{figure}

\subsection{Systematic effects}
\label{sec:systematics}

A reliable bound requires control of residual trap or resonator nonlinearities beyond polynomial order, dispersive shifts from the readout or ancilla used to resolve Fock states, drive-induced Stark shifts (which scale with probe power and can be extrapolated to zero), thermal occupation and imperfect Fock-state preparation, frequency pulling, and calibration drift during the time needed to scan $N$ transitions. The choice of $\delta$ also has to be justified for the platform. If each $\omega_n/\omega$ is known to a fractional precision $\varepsilon$ after these corrections, then $|\beta|\lesssim k\,\varepsilon$, with $k$ taken from Table~\ref{tab:fisher}. We do not claim a particular $\varepsilon$ for any existing apparatus.

\subsection{Withdrawn bounds}
\label{sec:withdrawn}

An earlier version of this paper estimated $|\beta|\lesssim10^{-5}$ from superconducting-qubit coherence times and $|\beta|\lesssim10^{-9}$ from Ramsey frequencies. Both estimates are withdrawn. The first treated a coherent frequency shift as an exponential decay rate and used echo coherence times, which are insensitive to static dephasing. The second compared a fitted transition frequency with a supposedly predicted one. Both also used $|\log(E/E_*)+1|\sim20$, which corresponds to $E_*\sim1\,\text{s}^{-1}$ and not to the GHz-scale $E_*$ stated in the text.

\section{Relation to Other Work}
\label{sec:related}

\paragraph{Modified dispersion relations.} The closest related work concerns changes to the relation between energy or momentum and phase. Quantum-gravity phenomenology has studied dispersion relations of the form $E^2=p^2+m^2+\eta\,p^3/E_{\text{P}}+\cdots$ in detail \cite{AmelinoCamelia2013}. Applying $E\mapsto F_\beta(E)$ to single-particle energies gives a logarithmic dispersion relation of this general type, and existing time-of-flight and interference constraints would apply to it once the origin is fixed. For free particles the natural origin is the rest energy, which is the absolute convention, and laboratory curvature effects are then very small.

\paragraph{Operational time.} If a clock runs at a rate that depends on the energy of the system it times, $dt_{\text{phys}}/dt_{\text{ideal}}=1+\alpha\log(E/E_0)$, then evolution in ideal time has the form of Eq.~\eqref{eq:eigenstate} with $\beta=\alpha$ and $E_*=E_0$. This requires a particular coupling between clock and system, which is an additional assumption. In the clock-based work of Gambini, Porto and Pullin \cite{Gambini2007}, the relevant effect is clock fluctuations, which lead to non-unitary evolution of Lindblad type; $F_\beta(H)$ would describe only a deterministic mean rate.

\paragraph{Logarithms in effective actions.} Running couplings $g(E)=g_0+\alpha\log(E/\mu)$ and operators such as $\log(\Box/\mu^2)$ in semiclassical gravity \cite{Birrell1982,Parker2009} show that logarithms of energy are common, but they do not lead to Eq.~\eqref{eq:deformation}. In quantum field theory, running is absorbed into the measured parameters, and the resulting levels are simply the eigenvalues of $H$. An imaginary term $i\hbar\gamma\log(E/E_{\text{P}})$ in a semiclassical action gives a real amplitude factor $(E/E_{\text{P}})^{-\gamma}$, and a real term gives a phase $(E/E_{\text{P}})^{i\gamma}$ that does not grow with $t$. These examples motivate the functional form but do not derive it.

\paragraph{Logarithmic nonlinear Schrödinger equation.} Bia{\l}ynicki-Birula and Mycielski \cite{BBM1976} added a term $-b\log(|\psi|^2a^3)\psi$ to the Schrödinger equation. Their modification is nonlinear in the state, whereas $F_\beta(H)$ is linear in the state and logarithmic in the energy, so the neutron-interferometry bounds on their model \cite{Shull1980,Gahler1981} do not apply to $\beta$. For the same reason $F_\beta(H)$, which has no memory kernel, is unrelated to fractional Schrödinger dynamics \cite{Oldham2006,Podlubny1999}.

\paragraph{Decoherence models.} Collapse models, including GRW and CSL \cite{Ghirardi1986,Bassi2003,Bassi2013}, gravitational collapse \cite{Diosi1987,Penrose1996} and energy-conserving variants \cite{Snoke2021,Snoke2023a,Snoke2023b,Snoke2024}, as well as intrinsic-decoherence models \cite{Milburn1991}, environment-induced decoherence \cite{Breuer2002,Caldeira1981} and non-Hermitian deformations \cite{Matsoukas2023a,Matsoukas2023b} produce non-unitary reduced or ensemble dynamics, loss of purity, or non-Hermitian conditional evolution. The present model does none of these. The loss of coherence it can produce is like wave-packet dephasing in neutrino oscillations \cite{Giunti1998}: static, deterministic and in principle reversible.

\section{Limitations and Outlook}
\label{sec:discussion}

The constant-$\beta$ logarithmic deformation is a phenomenological model of a nonlinear energy--phase relation, and its predictions hold only under the stated conventions. Its dependence on $E_*$ amounts to a rescaling of the spectrum. The part that is independent of $E_*$ appears in second divided differences, which sample the curvature of the spectral map, and the most direct test is the $1/(n+\delta)$ anharmonicity of a nominally harmonic ladder, measured over enough levels to separate it from conventional nonlinearities.

Whether $\beta$ is a physical parameter or only a reparametrization depends on three open problems. First, second differences remove the dependence on $E_*$ but not on the energy origin, and the ansatz does not select an origin; in the absolute convention the effects are unobservably small. Second, $F_\beta$ is not additive over subsystems, and the factor-wise rule depends on how the system is divided and on when interactions are switched on. For the same reason the analysis here is limited to a single system during free evolution. Third, $\log H$ depends on the reference frame. A covariant deformation of single-particle dispersion relations, with the composite-system rule obtained from second quantization and the origin set by rest mass, might resolve all three and would connect the model to existing dispersion constraints. In the absolute convention, however, it would also make laboratory effects far too small to detect. Continuous spectra and scattering problems would instead sample the curvature over a range of energies.

The choice $G(E)=E\log(E/E_*)$ is made because it is the simplest function whose curvature, $G''=1/E$, has no intrinsic scale. The divided-difference and ladder analysis applies to any $F(E)$ through its second derivative.

\section*{Acknowledgments}

The author thanks colleagues in quantum foundations and quantum gravity for discussions, and readers of earlier versions whose comments improved the treatment of the energy origin.

\appendix

\section{Non-Stationary-Phase Estimate}
\label{app:theorem}

\begin{theorem}
\label{thm:nonstationary}
Let $f \in C^1_c((0, \infty))$ with $\mathrm{supp}\, f \subset [E_{\min}, E_{\max}]\subset(0,\infty)$, $t \neq 0$, and $I_\beta(t)=\int f(E)e^{-itF_\beta(E)}dE$. If there is $m>0$ with $\abs{1 + \beta(\log(E/E_*) + 1)} \geq m$ on $[E_{\min},E_{\max}]$, then
\begin{equation}
|I_\beta(t)| \leq \frac{2\,\|f'\|_{L^1}}{m\,|t|}.
\label{eq:decay}
\end{equation}
\end{theorem}

\begin{proof}
With $\Phi_\beta=tF_\beta$ and $f$ vanishing at the endpoints, integration by parts gives $I_\beta=\int \frac{d}{dE}[f/(i\Phi'_\beta)]e^{-i\Phi_\beta}dE$, so
$|I_\beta|\le\int|f'|/|\Phi'_\beta|+\|f\|_\infty\int|\Phi''_\beta|/\Phi'^2_\beta$. By hypothesis and continuity, $\Phi'_\beta$ has constant sign on $[E_{\min},E_{\max}]$ with $|\Phi'_\beta|\ge m|t|$, and $\Phi''_\beta=t\beta/E$ also has constant sign, so $1/\Phi'_\beta$ is monotone and the second integral equals $|1/\Phi'_\beta(E_{\min})-1/\Phi'_\beta(E_{\max})|$. The two reciprocals have the same sign and modulus at most $1/(m|t|)$, so both lie in $[0,1/(m|t|)]$ or both lie in $[-1/(m|t|),0]$; their difference therefore has modulus at most $1/(m|t|)$. Since $\|f\|_\infty\le\|f'\|_{L^1}$, the claim follows.
\end{proof}

At $\beta=0$ we have $m=1$, and the same $\order{1/|t|}$ decay holds; this is the Riemann--Lebesgue behavior of any smooth wave packet. For $|\beta|\ll1$ the deformation changes the bound only through $m=1+\order{\beta}$. If $|\beta|\to\infty$ and the support avoids $E_*/e$, then $|I_\beta|=\order{1/(|\beta||t|)}$, but that limit has no physical relevance. The estimate does not point to a new source of decay.


\begin{thebibliography}{99}

\bibitem{AmelinoCamelia2013}
G.~Amelino-Camelia,
``Quantum-spacetime phenomenology,''
Living Rev.\ Relativ.\ \textbf{16}, 5 (2013).

\bibitem{Gambini2007}
R.~Gambini, R.~A.~Porto, and J.~Pullin,
``Fundamental decoherence from quantum gravity: a pedagogical review,''
Gen.\ Relativ.\ Gravit.\ \textbf{39}, 1143--1156 (2007).

\bibitem{Birrell1982}
N.~D.~Birrell and P.~C.~W.~Davies,
\textit{Quantum Fields in Curved Space},
Cambridge University Press, Cambridge (1982).

\bibitem{Parker2009}
L.~Parker and D.~Toms,
\textit{Quantum Field Theory in Curved Spacetime},
Cambridge University Press, Cambridge (2009).

\bibitem{BBM1976}
I.~Bia{\l}ynicki-Birula and J.~Mycielski,
``Nonlinear wave mechanics,''
Ann.\ Phys.\ \textbf{100}, 62--93 (1976).

\bibitem{Shull1980}
C.~G.~Shull, D.~K.~Atwood, J.~Arthur, and M.~A.~Horne,
``Search for a nonlinear variant of the Schrödinger equation by neutron interferometry,''
Phys.\ Rev.\ Lett.\ \textbf{44}, 765--768 (1980).

\bibitem{Gahler1981}
R.~Gähler, A.~G.~Klein, and A.~Zeilinger,
``Neutron optical tests of nonlinear wave mechanics,''
Phys.\ Rev.\ A \textbf{23}, 1611--1617 (1981).

\bibitem{Oldham2006}
K.~Oldham and J.~Spanier,
\textit{The Fractional Calculus},
Dover Publications, Mineola, NY (2006).

\bibitem{Podlubny1999}
I.~Podlubny,
\textit{Fractional Differential Equations},
Academic Press, San Diego (1999).

\bibitem{Ghirardi1986}
G.~C.~Ghirardi, A.~Rimini, and T.~Weber,
``Unified dynamics for microscopic and macroscopic systems,''
Phys.\ Rev.\ D \textbf{34}, 470 (1986).

\bibitem{Bassi2003}
A.~Bassi and G.~C.~Ghirardi,
``Dynamical reduction models,''
Phys.\ Rep.\ \textbf{379}, 257--426 (2003).

\bibitem{Bassi2013}
A.~Bassi, K.~Lochan, S.~Satin, T.~P.~Singh, and H.~Ulbricht,
``Models of wave-function collapse, underlying theories, and experimental tests,''
Rev.\ Mod.\ Phys.\ \textbf{85}, 471--527 (2013).

\bibitem{Diosi1987}
L.~Diósi,
``A universal master equation for the gravitational violation of quantum mechanics,''
Phys.\ Lett.\ A \textbf{120}, 377--381 (1987).

\bibitem{Penrose1996}
R.~Penrose,
``On gravity's role in quantum state reduction,''
Gen.\ Relativ.\ Gravit.\ \textbf{28}, 581--600 (1996).

\bibitem{Snoke2021}
D.~W.~Snoke,
``A model of spontaneous collapse with energy conservation,''
Found.\ Phys.\ \textbf{51}, 100 (2021).

\bibitem{Snoke2023a}
D.~W.~Snoke,
``Mathematical formalism for nonlocal spontaneous collapse in quantum field theory,''
Found.\ Phys.\ \textbf{53}, 34 (2023).

\bibitem{Snoke2023b}
D.~W.~Snoke and D.~N.~Maienshein,
``Experimental predictions for norm-conserving spontaneous collapse,''
Entropy \textbf{25}, 1489 (2023).

\bibitem{Snoke2024}
D.~W.~Snoke,
\textit{Interpreting Quantum Mechanics: Modern Foundations},
Cambridge University Press, Cambridge (2024).

\bibitem{Milburn1991}
G.~J.~Milburn,
``Intrinsic decoherence in quantum mechanics,''
Phys.\ Rev.\ A \textbf{44}, 5401 (1991).

\bibitem{Breuer2002}
H.~P.~Breuer and F.~Petruccione,
\textit{The Theory of Open Quantum Systems},
Oxford University Press, Oxford (2002).

\bibitem{Caldeira1981}
A.~O.~Caldeira and A.~J.~Leggett,
``Influence of dissipation on quantum tunneling in macroscopic systems,''
Phys.\ Rev.\ Lett.\ \textbf{46}, 211--214 (1981).

\bibitem{Matsoukas2023a}
A.~S.~Matsoukas-Roubeas, F.~Roccati, J.~Cornelius, Z.~Xu, A.~Chenu, and A.~del Campo,
``Non-Hermitian Hamiltonian deformations in quantum mechanics,''
J.\ High Energ.\ Phys.\ \textbf{2023}, 060 (2023).

\bibitem{Matsoukas2023b}
A.~S.~Matsoukas-Roubeas, M.~Beau, L.~F.~Santos, and A.~del Campo,
``Unitarity breaking in self-averaging spectral form factors,''
Phys.\ Rev.\ A \textbf{108}, 062201 (2023).

\bibitem{Giunti1998}
C.~Giunti and C.~W.~Kim,
``Coherence of neutrino oscillations in the wave packet approach,''
Phys.\ Rev.\ D \textbf{58}, 017301 (1998).

\bibitem{Hahn1950}
E.~L.~Hahn,
``Spin echoes,''
Phys.\ Rev.\ \textbf{80}, 580--594 (1950).

\end{thebibliography}
\end{document}